\documentclass[a4,technote,11pt]{IEEEtran}\onecolumn

\usepackage{graphicx,epsf}

\newtheorem{theorem}{Theorem}

\newcounter{example}
\newenvironment{example}%
{\begin{trivlist}\refstepcounter{example}%
\item[]{\em Example {\em \theexample}}%
\nopagebreak[4]}%
{\mbox{}\hfill\QEDclosed
\end{trivlist}}
\renewcommand{\theexample}{\arabic{example}}

\newcommand{\txt}{\textstyle}

\newcommand{\wh}{\widehat}

\newcommand{\ulp}{{\rm ulp}}

\newcommand{\specialdescriptionlabel}[1]{\bf #1\hfil}

\begin{document}

\title{Correction to the 2005 paper:\\"Digit Selection for SRT Division
and Square Root''}


\author{Peter Kornerup\\
Dept.~of Mathematics and Computer Scienc\\
University of Southern Denmark, Odense, Denmark\\
{\em E-mail: kornerup@imada.sdu.dk}}%


\maketitle

\begin{abstract} 
  It has been pointed out by counterexamples in a 2013 paper in the IEEE
  Transactions on Computers \cite{Rus13}, that there is an error in the
  previously ibid.\ in 2005 published paper \cite{Kor05a} on the
  construction of valid digit selection tables for SRT type division and
  square root algorithms. The error has been corrected, and new results
  found on selection constants for maximally redundant digit sets.
\end{abstract}

\begin{keywords}
Digit~selection, SRT, division, square root
\end{keywords}~\\[-6ex]

\section{Introduction}
In a recent paper \cite{Rus13}, David M. Russinoff expressed criticism on
the determination of digit selection parameters in SRT division algorithms, presented in the paper \cite{Kor05a} by this author. An
error in the selection was pointed out by counterexamples.

The SRT algorithms for division and square root are based on selecting the
digits of the result by a table look-up or equivalent, using a few leading
bits of the divisor (or root approximation) and of the partial remainder.

To determine the minimal number of bits necessary for a valid table to
exist, traditionally searches were performed to assure that the next
quotient digit can be chosen as valid for all points (remainder, divisor)
in a set defined by the truncated remainder and divisor, i.e., a specific
“uncertainty rectangle.”

It was the purpose of \cite{Kor05a} as an alternative to present a more
analytical approach to determine these parameters, based directly on the
radix and digit set of the quotient/root representation. Below is a brief
account of the core of this approach, but with some additional
considerations on these parameters, followed by an analysis of the error
in \cite{Kor05a}, including a correction for its parameter determination
and simplifications when the digit set is maximally redundant. 
Finally some conclusions are presented.~\\[-4ex]

\section{Parameter determination in \cite{Kor05a}}

Let $\beta$ be the radix (assumed to be a power of 2) and $\{-a,..,a\}$
be the digit set of the quotient, where $\beta/2\le a \le \beta-1$ and
$\rho=\frac{a}{\beta-1}$ is the redundancy index. Let $t$ and $u$ to be
determined be respectively the number of leading fractional digits of the
partial remainder and of the divisor $y$, assumed normalized $1/2\le y
<1$. The digit selection is to be based on a table look-up (or equivalent)
essentially indexed by $t$ and $u$, the table size then being exponential
in $u+t$. Hence we seek $t$ and $u$ such that $u+t$ is minimal, normally
obtained by minimizing $u$, which in \cite{OF98} by synthesis studies has
been confirmed also to generally minimize the delay and area. 

An analysis on the positioning of the ``uncertainty rectangles'' leads to
the following condition (Equation (12) from \cite{Kor05a}) for a valid
digit selection table to exist:
\begin{eqnarray}\nonumber
\lefteqn{\left\lceil 2^{t-u}(d-\rho)k+2^{t-u}(d-\rho)+1\right\rceil}\\
\label{sel-cond}
 & \le \left\lfloor2^{t-u}(d-\rho)k+2^{t-u}(2\rho-1)k\right\rfloor,
\end{eqnarray}
which has to be satisfied for all $k$, $2^{u-1}\le k < 2^u$ and digits
$d>0$ (which can be assumed by symmetry). Note that the inner leftmost
terms in the two sides of the inequality are identical, and thus the
condition essentially depends on the rightmost terms. Also note that the
significance of the terms is increased by maximizing the difference $t-u$.

It is then seen that (\ref{sel-cond}) is satisfied if the following:
\[
2^{t-u}\left((2\rho-1)k - (d-\rho)\right) \ge 2,
\]
holds for the minimal value $k=2^{u-1}$ and the maximal value $d=a$, and
thus also for all $d < a$ and $k > 2^{u-1}$. 
This translates into the condition:
\begin{equation}\label{strong}\txt
2^{-t}\le \left(\left(\rho-\frac{1}{2}\right)-(a-\rho)2^{-u}\right)/2
\end{equation}
which may be used to find values of $u$ and $t$. However, there is a chance
that the inequality (\ref{sel-cond}) can be satisfied even if the weaker
condition:
\begin{equation}\label{weak}\txt
2^{-t}\le \left(\left(\rho-\frac{1}{2}\right)-(a-\rho)2^{-u}\right)
\end{equation}
similarly is satisfied. 

For any of these conditions to hold, it is obviously necessary that $u$ is
chosen such that:
\begin{equation}\label{eq-u-cond}
2^{-u}<\frac{\rho-\frac{1}{2}}{a-\rho},
\end{equation}
provided that $a>\rho$, or $\beta>2$, since $\beta=2$ is the only case
where $\rho=a(=1)$, a case which can be handled separately.

Given any value of $u$ satisfying (\ref{eq-u-cond}), a possible value
$t=t_0$ can then  be determined say from (\ref{weak}) as:
\begin{equation}\label{t-expr}\txt
t_0 = 
\left\lceil\left(-\log_2(\rho-\frac12-(a-\rho)2^{-u}\right)\right\rceil,
\end{equation}
however, it may be necessary to apply the stronger condition
(\ref{strong}), in which case $t=t_0+1$. To decide between these two
situations the difference between the righthand and lefthand expressions in
(\ref{sel-cond}) may be checked for given specific values of $u$ and $t_0$.

Note that $u$ by (\ref{eq-u-cond}) can be chosen arbitrarily large, and
that $t_0$ by (\ref{t-expr}) decreases when $u \rightarrow \infty$, e.g.,
for maximally redundant digit sets ($\rho=1$), $t_0\rightarrow 2$. Hence
the factor $2^{t-u}$ in (\ref{sel-cond}) can be made arbitrarily small.
However, we want $u+t$ to be small to minimize the table.~\\[-4ex]

\section{The error and its correction}

Russinoff in \cite{Rus13} points out by counterexamples for large values of
the radix and $u$, that the test in \cite{Kor05a}, fails in some cases to
correctly identify whether to use $t=t_0$ or $t=t_0+1$.

The test is based on checking whether the difference $\Delta(t,u,d,\rho,k)$
between the expressions in (\ref{sel-cond}) is non-negative for $2^{u-1}\le
k <2^u$, together with the (false!) observation that it is sufficient to
perform it only for the digit value $d=a$, to assure that the inequality
holds for all values of $d>0$. If tested for all values of $d$ it is
equivalent to a check on the correct positioning of the ``uncertainty
rectangles'' between some slanted lines. The counterexamples were found for
radix 16 and 32 with $u=9$, respectively $u=11$ and $t_0=2$, which
erroneously had accept for $d=a$ but failed for $d=a-1$. Note that in both
cases the value of $2^{t-u}$ is very small.

Let $\delta_{kd}=2^{t_0-u}\left((2\rho-1)k-(d-\rho)\right)-1$ be the
difference between the internal expressions in (\ref{sel-cond}). The test
according to Theorem~3 of \cite{Kor05a} fails for certain extreme
combinations of $u$ and $t_0$ ($u \gg t_0)$, since the determination of $t$
from (\ref{t-expr}) does not assure that $\delta_{kd}\ge1$. When
$\delta_{kd}<1$ the above observation on the sufficiency of the test for
$d=a$ does not hold. Note that for $u \gg t_0$, $\delta_{kd}$ grows only
slowly with~$k$.  In the two counterexamples with $t_0=2$ it is found that:
\[\arraycolsep 3pt
\begin{array}{ccccl}
\beta=16 & u=9  & d=a=15 & k=2^{u-1} & \delta_{kd}=0.890625\\
\beta=32 & u=11 & d=a=31 & k=2^{u-1} & \delta_{kd}=0.94140625.
\end{array}
\]

If $0<\delta_{kd}<1$ and the two internal expressions in (\ref{sel-cond})
considered as an interval happens to include an integer value $x$ for some
value of $k$, then $\Delta(t,u,d,\rho,k)=0$. But for each increment of $k$,
the left endpoint of the interval will be shifted to the right by an amount
$2^{t_0-u}(d-\rho)$ and the width increased by $2^{t_0-u}(2\rho-1)$.
Eventually, if still $\delta_{kd}<1$, it may fall in the open interval
between $x$ and $x+1$, then $\Delta(t,u,d,\rho,k)=-1$ and the test
fails. If this happens for $d=a$ then Theorem~3 in \cite{Kor05a} specifies
that $t=t_0+1$ should be used.

If a smaller value of $t$ is wanted, a value larger than the minimal values
of $u$ may be used, not necessarily minimizing $u+t$. From the stronger
condition (\ref{strong}), for any value of $t
\ge\lceil-\log_2(\rho-\frac12)\rceil+1$, we find that
\begin{equation}\label{eq-u'}
2^{-u'} \le \frac{\rho-\frac12-2^{1-t}}{a-\rho},
\end{equation}
implies $\delta_{kd} \ge 1$ for $k=2^{u'-1}$, and hence
$\Delta(t,u,d,\rho,k)\ge 0$ for all $k>2^{u'-1}$ and $d\le a$. From
(\ref{eq-u'}) we may determine a $u'$ which may be greater than the minimal
$u$ chosen by (\ref{eq-u-cond}).

Thus in the above counterexamples, $t=3$ is the minimal value possible for
$\beta=16$ and $\beta=32$. Then $u'=6$ respectively $u'=7$ are the minimal
values which could be used, and choosing these values we find:
\[\arraycolsep 3pt
\begin{array}{ccccl}
 \beta=16 & u'=6 & d=a=15 & k=2^{u'-1} & \delta_{kd}=1.25\\
 \beta=32 & u'=7 & d=a=31 & k=2^{u'-1} & \delta_{kd}=1.125,
\end{array}
\]
where the test accepts.

\section{The correction to \cite{Kor05a}}
Theorem~3 in \cite{Kor05a} is based on the chance that the weaker condition
(\ref{weak}) is sometimes sufficient to satisfy condition (\ref{sel-cond}),
and thus the smaller value $t=t_0$ can be used. But exhaustive searches for
$u=u_0$ being the minimal solution to (\ref{eq-u-cond}), has shown that
this turns out to be the case in only very few cases. However, it has
turned out that no searches are necessary in the case when the digit set is
maximally redundant, hence we will deal with this case separately below.

It turns out that the test for all $d$ whether $t_0=\hat{t}$ on
$\Delta(\hat{t},u_0,d,\rho,k)$ for a restricted set of radices is only
satisfied for $\beta=4, a=2$, for $\beta=16, a=10$, for $\beta=32, a=25$,
for $\beta=64, a=38,42,44,46,51$, and for $\beta=128, a=81,89,94,105$. In a
few other cases the tests falsely indicated accept. The test on
$\Delta(\hat{t},u_0,d,\rho,k)$ generally fails when $2^{\hat{t}-u_0}$ is very
small, but in no systematic way.  A Maple program for the general
determination of valid parameters is available.

A corrected and reorganized version of Theorem~3 of \cite{Kor05a}, now
identifying further valid parameter pairs $(u,t)$, is then:

\begin{theorem}{\em (SRT digit selection constants)}
  \label{thm}\\
  For $p$-bit radix $\beta$ SRT division for $\beta=2^p, p=2,.., 7$ with
  digit set $D=\{-a,..,a\}$, $\beta/2 \le a < \beta-1$, and
  $\rho=\frac{a}{\beta-1}$, the selection constants
  $\wh{S_d}(\wh{y})=s_{d,k}2^{-t}$ can be determined for $1\le d\le a$ and
  $\wh{y}=k\cdot\ulp(\wh y)$ as
\[
s_{d,k}=\left\lceil2^{t-u}(d-\rho)(k+1)\right\rceil
\]
for $k=2^{u-1},..,2^u-1$, using truncation parameters $t,u$ defined by
$\ulp(\wh{S_d}(\wh{y}))=\ulp(\wh{\beta r}_i) = 2^{-t}$ and
$\ulp(\wh{y})=2^{-u}$, where $u$ has to satisfy\\[-3ex]
\begin{equation}
\label{eq-u}
2^{-u}<\frac{\rho-\frac{1}{2}}{a-\rho}.
\end{equation}

If $u=u_{min}$ is the minimal value satisfying (\ref{eq-u}), then let
$t'$ be the smallest value of $t$ satisfying\\[-2ex]
\begin{equation}
\label{t-cond}
\txt
t > 1-\log_2(\rho-\frac{1}{2}),
\end{equation}
and define $u=u_{max}$ as the smallest value of $u$ satisfying
\begin{equation}\label{u'-cond}
2^{-u} \le \frac{\rho-\frac12-2^{1-t'}}{a-\rho}.
\end{equation}

For any value of $u$, $u_{min} \le u \le u_{max}$ define $\hat{t}$ as the
smallest value of $t$ satisfying
\begin{equation}\txt
2^{-t} \le (\rho-\frac12)-(a-\rho)2^{-u},
\end{equation}
and define from (\ref{sel-cond})\\[-4ex]
\begin{eqnarray*}
\Delta(t,u,d,\rho,k)&=&
\left\lfloor2^{t-u}(d+\rho-1)k\right\rfloor
-\left\lceil 2^{t-u}(d-\rho)(k+1))+1\right\rceil.
\end{eqnarray*} 

Also define the following two checks:
\[simple=
\exists k \in \{2^{u-1}\ldots2^u-1\}\,:\,\Delta(\hat{t},u,a,\rho,k)<0
\]
and
\begin{eqnarray*}
{rest = \exists k \in \{2^{u-1}\ldots2^u-1\} \mbox{ and } \exists d \in
  \{0\ldots a-1\}}\\
: \Delta(\hat{t},u,d,\rho,k)< 0
\end{eqnarray*}

Then
\[
t=\left\{
 \begin{array}{ll}
\hat{t}+1 &{\bf if }\,\,  simple,\\
\hat{t}+1 &{\bf if }\,\, \neg simple \wedge rest,\\
\hat{t} &{\bf otherwise,}
\end{array}
\right.
\]
then $(u,t)$ provides a set of parameters defining a valid digit selection
table.
\end{theorem}

\begin{proof}
  The expression for $s_{d,k}$ is from (\ref{sel-cond}), and the condition
  (\ref{eq-u}) on $u$ is necessary, from which the minimal value $u_{min}$
  is derived. Comparing (\ref{eq-u}) with (\ref{u'-cond}) it is seen that
  $u_{max}\ge u_{min}$,
   
  The only situations where $t\!=\!\hat{t}$ can be verified, given $\beta$
  and $a$, are when $\Delta(\hat{t},u,d,\rho,k)\ge0$ for all
  $d\in\{1,\cdots,a\}$ and $k\in\{2^{u-1},\cdots,2^u-1\}$, yielding the
  combinations listed. The split cases when $\hat{t}$ must be increased
  covers situations where the test fails for some value of $d$ and $k$, and
  the strong condition (\ref{strong}) must be applied.
\end{proof}
It is only necessary to check if $\Delta(\hat{t},u,d,\rho,k)\ge0$ for
$d\in\{1\cdots,a-1\}$ and all $k$ when $\Delta(\hat{t},u,d,\rho,k)\ge~0$
for all $k$. This is where the original theorem failed by only testing the
latter. But observe that if the simple test turns out false, no further
testing is necessesary. As mentioned above there are only very few
situations where $t=\hat{t}$.

Also note that no solutions are possible for $u<u_{min}$, and if $(u,t)$ is
a valid pair, then $(u+s,t)$ and $(u,t+s)$ for any $s>0$ are also, but
obviously not as good.
\begin{example}
  With the minimally redundant digit set for $\beta=16$, $a=8$, $u_{min}=8$
  and $u_{max}=12$. For the possible choices of $u$ we find:
  \[
\begin{array}{rrc}
   u  & t & u+t \\
 8 & 9 & 17 \\
 9 & 7 & 16 \\
10 & 7 & 17 \\
11 & 7 & 18 \\
12 & 6 & 18 
\end{array}
\]
where $(u,t)=(9,7)$ yields the minimal value of $u+t$.
\end{example}

\begin{theorem}{\em (SRT for maximally redundant digit sets)}
  \label{thm2}\\
  For $\beta=2^p$, $p>2$, with the maximally redundant digit set
  $D=\{-\beta+1\cdots,0,\cdots,\beta-1\}$ there are two sets of parameters
  $(u,t)$ defining valid digit selection tables:
\[
\begin{array}{ccc}
 u & t & u+t\\
 u_{min}=p+1 & p & 2p+1\\
 u_{max}=p+2 & 3 & p+5
\end{array}
\]
For $p=2$, $u=u_{min}=u_{max}$ there is only one set:
\[
\begin{array}{ccc}
 u & t & u+t\\
 3 & 2 & 5.
\end{array}
\]
\end{theorem}
\begin{proof}
  With $a=\beta-1$, $\rho=1$, for $\beta=2^p$ it follows that $u_{min} =
  \lceil\log_2(2^p-2)+1\rceil = p+1$. Then $t'=3$, from which $2^{u_{max}}
  \ge 2^{p+2}-8$, implying $u_{max}=p+2$ when $p>2$, but for $p=2$
  $u_{max}=3$, which is identical to $u_{min}$.  From
  $2^{-t}\le(\rho-\frac12)-(a-\rho)2^{-u}$ for $u=u_{min}=p+1$ the minimal
  $t$ is $\hat{t}_{min}=p$, and for $u=u_{max}=3$, $\hat{t}_{max}=2$.

  For $(u,t)=(u_{min},\hat{t}_{min})=(p+1,p)$:
\begin{eqnarray*}
\Delta(\hat{t}_{min},u_{min},d,1,k)
&\hspace{-15ex}=\left\lfloor2^{-1}d\,k\right\rfloor
-\left\lceil 2^{-1}dk+2^{-1}(d-k-1)+1\right\rceil\\
&\ge \left\lfloor2^{-1}d'k'\right\rfloor
-\left\lceil 2^{-1}d'k+2^{-1}((2^p-1)-(2^p+1)-1)+1\right\rceil\\
&\hspace{-25ex}=\left\lfloor2^{-1}d'k'\right\rfloor
\txt-\left\lceil 2^{-1}d'k'-\frac12\right\rceil = 0,
\end{eqnarray*} 
when substituting $d$ by its maximal value $d'=2^p-1$ and $k$ by its
(almost) minimal value $k'=2^{u_{min}-1}+1=2^p+1$, using that with these
extreme values $d'k'$ is odd. Substituting with $k"=2^{u_{min}-1}=2^p$ then
$d'k"$ is even and the lower bound is $\lfloor 2^{-1}d'k"\rfloor -\lceil
2^{-1}d'k"\rceil = 0 $. 

Hence $(u,t)=(u_{min},\hat{t}_{min})=(p+1,p)$
provides a correct table, which also covers the case for $p=2$ with
$(u,t)=(3,2)$.

For $(u,t)=(u_{max},\hat{t}_{max})=(p+2,2)$ we will now show that there is a
value $k=2^{u_{max}-1}+2=2^{p+1}+2$ such that $\Delta(t,u,a,\rho,k)<0$ for
$d=a=2^p-1$.
Let $K=2^{\hat{t}_{max}-u_{max}}a\,k=2(2^p-2^{-p})$ then
\begin{eqnarray*}
\Delta(\hat{t}_{max},u_{max},a,1,k)
&&=\left\lfloor K\right\rfloor
-\left\lceil K+2^{-p}(a-k-1)+1\right\rceil\\
&&=\left\lfloor K\right\rfloor
-\left\lceil K+2^{-p}(2^p-1-2(1+2^p)+1)+1\right\rceil\\
&&=\left\lfloor K\right\rfloor
-\left\lceil K -2^{-p}\right\rceil\,=\,(2^{p+1}-1)-2^{p+1}\,=-1.
\end{eqnarray*} 
Thus $\hat{t}_{max}$ must be incremented and $(u,t)=(p+2,3)$ provides a
correct table.
\end{proof}

\begin{example}
  For the maximally redundant digit set with $\beta=16$, $a=15$, the
  minimal value of $u$ is $u_{min}=5$, for which $t'=\hat{t}=4$ is
  determined. However, a smaller value of $t$, $t"=3$ is also possible, for
  which $u_{max}=6$. Note that $u+t$ is the same for the two combinations,
  hence they require the same table sizes, but when $t$ is smaller, fewer
  bits of the redundant partial remainder need to be converted.
\end{example}

\section{Conclusions}
It is likely that the error had not been noticed because it was implicitly
assumed that minimal tables are wanted, as obtained by choosing minimal or
almost minimal values of $u$, and thus for small values of $u+t$ and small
values of $u-t$. Fortunately, the exposure of the error has prompted a
further analysis of the problem of determining additional value pairs
$(u,t)$, providing valid digit selection tables. The result on truncation
parameters $u,t$ for the more general case has been significantly
strengthened. A new theorem is presented, simplifying the parameter
determination in the important case when the digit set is maximally
redundant, eliminating all searching.

As pointed out by Russinoff, the error has gone unnoticed in the review
process, and subsequently by other referencing the paper. However it is
standard scientific knowledge, that any research result has to prove its
correctness through the ``time test'', i.e. that it can stand uncontested
through time. It is appreciated that his objections identified the problem,
and made it possible in this case to provide a correction of the presented
results.

He further contests the use of ``informal quasi-mathematical arguments'' as
opposed to a ``formal machine-checked proof'', such as he has applied to
his proofs in \cite{Rus13}, employing an ACL2 proof script which consists
of more than 800 lemmas, an impressive effort.

His approach is the same as in publications before \cite{Kor05a} to
determine parameters for providing valid digit selection tables: proving
the validity of some pair $(u,t)$ by checking all table entries, but
limited to maximally redundant digit sets. The attempt in \cite{Kor05a} was
to determine these parameters directly, based on the radix and any
corresponding valid digit set, and in this revision has been significantly
strengthened.

{
\bibliographystyle{IEEEtran} 


\begin{thebibliography}{1}

\providecommand{\url}[1]{#1}
\csname url@rmstyle\endcsname
\providecommand{\newblock}{\relax}
\providecommand{\bibinfo}[2]{#2}
\providecommand\BIBentrySTDinterwordspacing{\spaceskip=0pt\relax}
\providecommand\BIBentryALTinterwordstretchfactor{4}
\providecommand\BIBentryALTinterwordspacing{\spaceskip=\fontdimen2\font plus
\BIBentryALTinterwordstretchfactor\fontdimen3\font minus
  \fontdimen4\font\relax}
\providecommand\BIBforeignlanguage[2]{{%
\expandafter\ifx\csname l@#1\endcsname\relax
\typeout{** WARNING: IEEEtran.bst: No hyphenation pattern has been}%
\typeout{** loaded for the language `#1'. Using the pattern for}%
\typeout{** the default language instead.}%
\else
\language=\csname l@#1\endcsname
\fi
#2}}

\bibitem{Rus13}
D.~Russinoff, ``{Computation and Formal Verification of SRT Quotient and Square
  Root Digit Selection Tables},'' \emph{IEEE Transactions on Computers},
  vol.~62, no.~5, pp. 900--913, May 2013.

\bibitem{Kor05a}
P.~Kornerup, ``{Digit Selection for SRT Division and Square Root},'' \emph{IEEE
  Transactions on Computers}, vol.~54, no.~3, pp. 294--303, March 2005.

\bibitem{OF98}
S.~Oberman and M.~Flynn, ``{Minimizing the Complexity of SRT Tables},''
  \emph{IEEE Transactions on VLSI systems}, vol.~6, no.~1, pp. {141--149},
  March 1998.

\end{thebibliography}

}

\end{document}